\documentclass[a4paper]{article}
\usepackage{amsmath,amsfonts,amsthm,amssymb}
\usepackage[aligntableaux=top]{ytableau}
\usepackage{tikz-cd}
\usetikzlibrary{matrix,arrows,calc}
\usepackage{calc}
\usepackage{hyperref}

\allowdisplaybreaks[1] 

\newtheorem{thm}{Theorem}
\newtheorem{prop}[thm]{Proposition}
\newtheorem{lem}[thm]{Lemma}
\newtheorem{cor}[thm]{Corollary}
\theoremstyle{remark}
\newtheorem{rem}{Remark}

  \newcommand{\yt}{\ytableaushort}
  \newcommand{\yd}{\ydiagram}
  
  \newcommand{\oo}{\infty}
  \newcommand{\oast}{\circledast}
\renewcommand{\d}{\mathrm{d}}
  \newcommand{\im}{\operatorname{im}}
  \newcommand{\tr}{\operatorname{tr}}
  \newcommand{\sso}{\subset}
  \newcommand{\sse}{\subseteq}
  \newcommand{\Secs}{\Gamma}
  
  \newcommand{\K}{\mathcal{K}}
  \newcommand{\KY}{\mathcal{KY}}
  \newcommand{\Lie}{\mathcal{L}}
  \newcommand{\G}{\mathrm{G}}
  \newcommand{\GL}{\mathrm{GL}}
  \newcommand{\R}{\mathbb{R}}

  \newcommand{\SC}{\mathit{sc}}
  \newcommand{\tc}{\mathit{tc}}
  \newcommand{\fc}{\mathit{fc}}
  \newcommand{\pc}{\mathit{pc}}
  \newcommand{\lf}{\mathit{lf}}
  \newcommand{\slf}{\mathit{slf}}
  \newcommand{\tlf}{\mathit{tlf}}
  \newcommand{\flf}{\mathit{flf}}
  \newcommand{\plf}{\mathit{plf}}

  \newcommand{\CMan}{\mathfrak{CMan}}

\title{Cohomology with causally restricted supports}
\author{Igor Khavkine\\
Department of Mathematics, University of Trento \\
and TIFPA-INFN, Trento\\
Via Sommarive, 14,
38123 Povo (TN) Italy\\ \texttt{igor.khavkine@unitn.it}}

\begin{document}
\maketitle


\begin{abstract}
De~Rham cohomology with spacelike compact and timelike compact supports has
recently been noticed to be of importance for understanding the structure of
classical and quantum Maxwell theory on curved spacetimes. Similarly causally
restricted cohomologies of different differential complexes play a similar role
in other gauge theories.  We introduce a method for computing these causally
restricted cohomologies in terms of cohomologies with either compact or
unrestricted supports.  The calculation exploits the fact that the
de~Rham-d'Alembert wave operator can be extended to a chain map that is
homotopic to zero and that its causal Green function fits into a convenient
exact sequence.  As a first application, we use the method on the de~Rham
complex, then also on the Calabi (or Killing-Riemann-Bianchi) complex, which
appears in linearized gravity on constant curvature backgrounds.  We also
discuss applications to other complexes, as well as generalized causal
structures and functoriality.
\end{abstract}



\section{Introduction}\label{sec:intro}
Recently, a number of works on the structure of classical and quantum
field theory on curved
spacetimes~\cite{dl,sdh,bds,fl,benini,hs,kh-big,kh-peierls} have made
use of de~Rham cohomology with spacelike compact supports. It appears in
the characterizations of the center of Poisson (or quantum) algebra of
observables of the Maxwell field and also of the degeneracy of the
bilinear pairing between spacelike compactly supported solutions and
compactly supported smearing functions (see
Proposition~\ref{prp:maxwell} for a specific statement). Similar
considerations appear in more general field
theories~\cite{kh-peierls,kh-big}, though involving cohomologies of
complexes that are different from the de~Rham one. One example is the
Calabi complex, which appears in linearized gravity on constant
curvature backgrounds~\cite[Sec.4.4]{kh-peierls} (see
Proposition~\ref{prp:lingrav} for a specific statement). Note that
cohomologies with timelike compact supports as well as on-shell
cohomologies (restricted to solution spaces of some particular
hyperbolic differential operators) have also appeared in the same
contexts. We shall loosely refer to all of these variations as
\emph{causally restricted cohomologies} or \emph{cohomologies with
causally restricted supports}.

It was noticed long ago~\cite{as} that non-trivial spacetime topology
can influence in a non-trivial way the construction of the classical and
quantum field theories. However, these effects had not been
systematically investigated until recently. This may explain why neither
the standard literature on differential geometry and topology, nor the
literature on relativity seem to have considered%
	\footnote{A notable exception is~\cite{lr}, which, as a byproduct of a
	different investigation, computed a few low degree cohomology groups
	with spacelike compact supports or restricted to solutions of the wave
	equation, but only on Minkowski space.} %
cohomologies with supports restricted by causal relations (like
spacelike or timelike compactness). So, given their growing importance,
they deserve independent investigation, which is the subject of this
work. We introduce a method that allows us to compute the causally
restricted cohomologies of a differential complex, provided that complex
is equipped with extra structure similar to that found in Hodge
theory~\cite{goldberg,jost}. The essentials of this method are
illustrated on the case of the de~Rham complex. Then, other applications
and implications are discussed.

In Section~\ref{sec:prelim}, we briefly outline some well known
geometric properties of the de~Rham complex on a Lorentzian spacetime,
as well as some basic facts of homological algebra. These properties
form the core of our method and are reminiscent of the structure found
in Hodge theory. Our method of computing causally restricted
cohomologies is then illustrated in Section~\ref{sec:comp} and is used
to express the various causally restricted de~Rham
cohomologies in terms of the standard de~Rham cohomologies with
unrestricted and compact supports.  Section~\ref{sec:calabi} applies the
same method to the Calabi differential complex. The Calabi complex plays
a role in linearized gravity on a constant curvature background
analogous to that of the de~Rham complex for Maxwell theory. Its
structure is briefly introduced and shown analogous to that highlighted
in Section~\ref{sec:prelim}. Then, in Section~\ref{sec:calabi-caus}, its
causally restricted cohomologies are computed in analogy with
Section~\ref{sec:comp}.  Section~\ref{sec:notes} discusses a few related
questions that have appeared in the study of gauge theories in the
framework of locally covariant classical and quantum field theory. In
particular, Sections~\ref{sec:gen-caus} and~\ref{sec:funct} deal with
the behavior of the causally restricted cohomology groups under changes
of causal structure and under embeddings, and Section~\ref{sec:other}
briefly describes how the methods applied to the de~Rham and Calabi
examples could be generalized to other differential complexes that arise
in the study of general field theories with constrains and gauge
invariance~\cite{kh-peierls,kh-big}. Finally, Section~\ref{sec:discuss}
concludes with a discussion of our results.

It should be mentioned that results very similar to those in
Section~\ref{sec:comp} have been obtained independently in the recent
work~\cite{benini}, though by a different methods. Those methods are
very specific to the de~Rham complex, including its invariance
properties under topological homotopies. Such strong invariance
properties certainly do not hold for other differential complexes. So it
is noteworthy that the content of our Sections~\ref{sec:calabi}
and~\ref{sec:notes} goes beyond~\cite{benini} in several directions.

\section{Preliminaries}\label{sec:prelim}
Fix an $n$-dimensional smooth manifold $M$ ($n\ge 2$) with a Lorentzian metric $g$
such that $(M,g)$ is an oriented, time-oriented space, globally
hyperbolic spacetime~\cite{wald,hawking-ellis,oneill,beem-ehrlich}. Recall
that, according to the \emph{Geroch splitting theorem}, there exists a
diffeomorphism $M \cong \R\times \Sigma$ (non-unique, of course) where
the corresponding projection $t\colon M \to \R$ is a \emph{Cauchy
temporal function}~\cite{geroch,bs1,bs2}.

Let $\Omega^p(M)$ denote the linear space of differential $p$-forms on
$M$ and let $\d \colon \Omega^p(M) \to \Omega^{p+1}(M)$ denote the
de~Rham differential, which together form the \emph{de~Rham complex}
\begin{equation}\label{eq:dr-cpx}
\begin{tikzcd}
	0 \arrow{r} &
	\Omega^0(M) \arrow{r}{\d} &
	\Omega^1(M) \arrow{r}{\d} &
	\cdots \arrow{r}{\d} &
	\Omega^n(M) \arrow{r} &
	0 ,
\end{tikzcd}
\end{equation}
This sequence of maps being a \emph{complex} means that each pair of
successive maps compose to zero, $\d \circ \d = 0$.
Its \emph{cohomology in degree $p$} is defined and denoted by
\begin{equation*}
	H^p(M) :=
		\frac{\ker (\d\colon \Omega^p(M) \to \Omega^{p+1}(M))}
			{\im (\d\colon \Omega^{p-1}(M) \to \Omega^p(M))} .
\end{equation*}
The cohomology of any other complex is defined in a similar way. It is
well known that this de~Rham cohomology is isomorphic, $H^p(M) \cong
H^p(M,\R)$, to the singular cohomology of $M$ with coefficients in
$\R$~\cite[Thm.15.8]{bott-tu}, to the \v{C}ech cohomology of $M$ with
coefficients in $\R$~\cite[Thm.8.9]{bott-tu}, and to the sheaf
cohomology of $M$ with coefficients in the sheaf of locally constant
$\R$-valued functions~\cite[Prp.10.6]{bott-tu}, all of which being
isomorphic are denoted by $H^p(M,\R)$. If we replace $\Omega^p(M)$
in~\eqref{eq:dr-cpx} with $\Omega_c^p(M)$, the linear space of
differential $p$-forms with compact support, the corresponding de~Rham
cohomology of $M$ with compact supports, which satisfies the following
isomorphism: $H_c^p(M)^* \cong H^p(M,\R)$. That isomorphism is
implemented by a non-degenerate bilinear pairing between
$\Omega^p(M)$ and $\Omega_c^{n-p}(M)$,
\begin{equation}
	\langle \alpha, \beta \rangle = \int_M \alpha\wedge \beta ,
\end{equation}
which descends to a non-degenerate bilinear pairing between $H^p(M)$ and
$H_c^p(M)$. This result is known as \emph{Poincar\'e
duality}~\cite[Rmk.5.7]{bott-tu}.

Using the Hodge star operator ${*}\colon \Omega^p(M) \to
\Omega^{n-p}(M)$ associated to the metric $g$, we can define the de~Rham
co-differential $\delta = {*}\d{*} \colon \Omega^p(M) \to
\Omega^{p-1}(M)$. Next, we define the so-called
\emph{de~Rham-d'Alembertian} or \emph{wave operator} $\square\colon
\Omega^p(M) \to \Omega^p(M)$,
\begin{equation}\label{eq:dal-def}
	\square = \d\delta + \delta \d .
\end{equation}
This operator differs from the simple tensor d'Alembertian $\nabla_a
\nabla^a$ by terms of lower differential order. From its very
definition, we see that the d'Alembertian is a cochain map from the
de~Rham complex to itself, $\d\square = \square\d$, which is moreover
cochain homotopic to zero, with the co-differential $\delta$ the
corresponding cochain homotopy. That is, it induces the zero map from
$H^p(M)$ to itself. The following diagram illustrates the discussion:
\begin{equation}\label{eq:dR-homotopy}
\begin{tikzcd}
	0 \arrow{r} &
	\Omega^0(M) \arrow{r}{\d} \arrow{d}{\square} &
	\Omega^1(M) \arrow{r}{\d} \arrow{d}{\square}
		\arrow[dashed]{dl}[swap]{\delta} &
	\cdots \arrow{r}{\d}
		\arrow[dashed]{dl}[swap]{\delta} &
	\Omega^n(M) \arrow{r} \arrow{d}{\square}
		\arrow[dashed]{dl}[swap]{\delta} &
	0 \\
	0 \arrow{r} &
	\Omega^0(M) \arrow{r}{\d} &
	\Omega^1(M) \arrow{r}{\d} &
	\cdots \arrow{r}{\d} &
	\Omega^n(M) \arrow{r} &
	0
\end{tikzcd} ,
\end{equation}
where the rows constitute (de~Rham) complexes, the solid arrows commute,
and the dashed arrows illustrate the cochain homotopy. This is an
important observation that will be used in an essential way in
Section~\ref{sec:comp}. Note that the formula~\eqref{eq:dal-def} is
analogous to the formula for the Hodge-de~Rham Laplacian in Riemannian
geometry. There, the observation that this Laplacian is homotopic to
zero lies at the foundation of Hodge theory~\cite{goldberg,jost}.

The causal structure on $M$ defined by the Lorentzian metric $g$ allows
us to restrict the supports of differential forms in other ways as well.
Recall that, for a subset $S \sse M$, by $J^\pm(S)$ we denote the subset
of $M$ that can be reached from $S$ by piecewise smooth, future ($+$) or
past ($-$) directed causal curves, while $J(S) = J^+(S) \cup J^-(S)$.
A closed set $S \sse M$ is said to be \emph{retarded} if $S \sse J^+(K)$
for some compact $K$, \emph{advanced} if $S\sse J^-(K)$ for some compact
$K$, \emph{spacelike compact} if it $S\sse J(K)$ for some compact $K$,
\emph{past compact} if $S\cap J^-(K)$ is compact for every compact $K$,
\emph{future compact} if $S\cap J^+(K)$ is compact for every compact
$K$, and \emph{timelike compact} if $S$ is both past and future
compact~\cite{sanders,baer}. Timelike compactness is also equivalent to
the property of having compact intersection with every spacelike compact
set. Let $\Omega^p_X(M)$, with $X=+,-,\SC,\pc,\fc$ or $\tc$, denote the
linear space of differential $p$-forms with, respectively, retarded,
advanced, spacelike compact, past compact, future compact or timelike
compact supports. For brevity, we refer to these spaces as space of
forms with \emph{causally restricted supports}.

Of course, since differential operators preserve supports, $\square$ also
restricts to $\square\colon \Omega_c^p(M) \to \Omega_c^p(M)$. By the
same reasoning, the spaces of forms with causally restricted supports
are also preserved by both $\d$ and $\square$. We define de~Rham
cohomology with causally restricted supports in the obvious way and
denote it by $H^p_X(M)$, with $X=+,-,\SC,\pc,\fc$ or $\tc$. Let
$\Omega^p_\square(M)$ and $\Omega^p_{\square,X}(M)$ denote the kernel of
the wave operator $\square$, also known as its \emph{solution space}, in
the spaces of forms with corresponding supports. Finally, by the
cochain map property, the de~Rham differential restricts to the kernel
of the wave operator, hence defining the de~Rham cohomology groups
$H^p_\square (M)$ and $H^p_{\square,X}(M)$ of solutions.

The specific way in which these causally restricted cohomologies are of
importance in Maxwell gauge theory is summarized in the following
proposition. For definiteness of notation let us fix a $\chi \in
C^\oo(M)$ that is $1$ in the future of a Cauchy surface $\Sigma_+$ and
$0$ in the past of another Cauchy surface $\Sigma_-$. The following is a
special case of the general result~\cite[Thm.3.2]{kh-peierls}.
\begin{prop}\label{prp:maxwell}
Maxwell gauge theory~\cite[Sec.4.2]{kh-peierls} induces a symplectic
form on $\Omega^1_{\square,\SC}(M)$~\cite[Def.3.10]{kh-peierls} that is
non-degenerate when (a) the bilinear form on $H^1_{\SC}(M) \times
H^{n-1}_c(M)$ induced by $\langle \alpha, \beta \rangle = \int_M \alpha
\wedge \beta$ is non-degenerate and (b) the bilinear form on
$H^1_{\square,\SC}(M)$ induced by $\langle \alpha, \beta
\rangle_{\square} = \int_M \alpha \wedge {*}\square (\chi \beta)$ is
non-degenerate (where $*$ denotes the Hodge dual).
\end{prop}
From the proof of that proposition it also follows that degeneracies in
(a) and (b) can imply degeneracies in the corresponding (pre-)symplectic
structure.

The wave operator on a globally hyperbolic Lorentzian manifold is well
known to be Green hyperbolic. That is, it has \emph{advanced} and
\emph{retarded} Green functions denoted respectively $\G_+$ and $\G_-$,
$\G_\pm \colon \Omega_c^p(M) \to \Omega_\pm^p(M)$. Since $\square$
commutes with $\d$, then so do $\G_+$ and $\G_-$. The form $\beta =
\G_\pm[\alpha]$ is the unique solution of $\square \beta = \alpha$ with,
respectively, retarded or advanced support. The domain of definition of
the Green functions can be extended, in a unique way, to $\Omega_X^p(M)$
for $X=+,-,\pc$ or $\fc$. Then, the maps
\begin{equation}\label{eq:biject}
	\square\colon \Omega_Y^p(M) \to \Omega_Y^p(M), \quad
	\G_X \colon \Omega_Y^p(M) \to \Omega_Y^p(M)
\end{equation}
are mutually inverse bijections, whenever $X=+$ and $Y=+$ or $\pc$, or $X=-$
and $Y=-$ or $\fc$. The combination $\G = \G_+ - \G_-$ is known as the
\emph{causal} Green function and fits into the following, in our
terminology Green-hyperbolic, exact
sequences~\cite{bgp,ginoux,kh-peierls,kh-big,baer}
\begin{equation}\label{eq:short-sc}
\begin{tikzcd}
	0 \arrow{r} &
	\Omega_c^p(M) \arrow{r}{\square} &
	\Omega_c^p(M) \arrow{r}{\G} &
	\Omega_{\SC}^p(M) \arrow{r}{\square} &
	\Omega_{\SC}^p(M) \arrow{r} &
	0 ,
\end{tikzcd}
\end{equation}
\begin{equation}\label{eq:short-tc}
\begin{tikzcd}
	0 \arrow{r} &
	\Omega_{\tc}^p(M) \arrow{r}{\square} &
	\Omega_{\tc}^p(M) \arrow{r}{\G} &
	\Omega^p(M) \arrow{r}{\square} &
	\Omega^p(M) \arrow{r} &
	0 .
\end{tikzcd}
\end{equation}
Note that, according to the above formulas, we can represent the space
of solutions with spacelike compact or unrestricted support either as
\begin{gather}
	\Omega^p_{\square,X}(M) = \ker \square \sso \Omega^p_X(M) \\
		\text{or}\quad
	\Omega^p_{\square,X}(M) = \G[\Omega^p_Y(M)] =
		\Omega^p_Y(M)/\square\Omega^p_Y(M),
\end{gather}
with $X=\SC$ and $Y=c$, or
$X$ empty and $Y=\tc$, respectively. On the other hand, we have trivial
solution spaces $\Omega^p_{\square,X}(M) = \{0\}$ when $X = +,-,\pc$ or
$\fc$.

The existence of the Green-hyperbolic exact sequences will allow us to
later make use of the following elementary result of homological
algebra~\cite[p.17]{bott-tu}. Let $A^\bullet = (A^p,\d)$ be a cochain
complex, and similarly for $B^\bullet$ and $C^\bullet$. It is well known
that a short exact sequence of cochain maps (maps commuting with the
differentials $\d$),
\begin{equation}
\begin{tikzcd}
	0 \arrow{r} & A^\bullet \arrow{r}{f} & B^\bullet \arrow{r}{g} & C^\bullet
	\arrow{r} & 0 ,
\end{tikzcd}
\end{equation}
induces a long exact sequence in cohomology,
\begin{equation}\label{eq:long-ex}
\begin{tikzcd}
	0 \arrow{r} &
	H^0(A^\bullet,\d) \arrow{r}{[f]} &
	H^0(B^\bullet,\d) \arrow{r}{[g]} \arrow[draw=none]{d}[name=Z,shape=coordinate]{}&
	H^0(C^\bullet,\d)
		\arrow[rounded corners,
			to path={ -- ([xshift=2ex]\tikztostart.east)
			|- (Z) [near end]\tikztonodes
			-| ([xshift=-2ex]\tikztotarget.west)
			-- (\tikztotarget)}]{dll}[above,pos=1]{[\d]} \\
	&
	H^1(A^\bullet,\d) \arrow{r}{[f]} &
	H^1(B^\bullet,\d) \arrow{r}{[g]} &
	H^1(C^\bullet,\d) \arrow{r}{[\d]} & \cdots
\end{tikzcd}
\end{equation}
The maps $[f]$, $[g]$ are induced by the corresponding cochain maps,
while the $[\d]$ maps are induced by the differentials of the complexes
(hence our notation for them) and are known as \emph{connecting
homomorphisms}.

\section{Computation of cohomology groups}\label{sec:comp}
In this section, we state and prove our main results on de~Rham
cohomology with causally restricted supports. We rely essentially on the
properties of the wave operator and its Green functions, as summarized
in Section~\ref{sec:prelim}. The important properties are that the wave
operator $\square$ is cochain homotopic to zero, and the way its range
and kernel are characterized using the causal Green function $\G$. In
particular, we do not explicitly rely on the invariance properties of
the de~Rham complex under topological homotopies.

\begin{thm}\label{thm:pf}
De~Rham cohomology $H^p_X(M)$, with $X=+,-,\pc$ or $\fc$, is trivial.
\end{thm}
\begin{proof}
Let $X=+,-,\pc$ or $\fc$. Then, as was noted in Section~\ref{sec:prelim},
the wave operator is a cochain map of the corresponding de~Rham complex
into itself, is invertible [Equation~\eqref{eq:biject}] and cochain
homotopic to zero [Equation~\eqref{eq:dal-def}]. Thus, it induces a map
in cohomology that is both invertible and equal to zero, which can only
mean that all the cohomologies are trivial. More concretely, given any
closed $\alpha \in \Omega^p_X(M)$, the identity $\d (\delta \G_X[\alpha])
= \G_X [(\d\delta + \delta\d) \alpha] = \alpha$ shows that it is also
exact.
\end{proof}

\begin{thm}\label{thm:sc-tc}
We have the isomorphisms
\begin{align}
	H^p_{\SC}(M) &\cong H^{p+1}_c(M), &
		H^p_{\square,\SC} &\cong H^p_c(M) \oplus H^{p+1}_c(M), \\
	H^p_{\tc}(M) &\cong H^{p-1}(M), & \text{and} ~
		H^p_{\square}(M) &\cong H^p(M) \oplus H^{p-1}(M),
\end{align}
with the convention that all cohomologies vanish in degree $p$ for $p<0$
or $p>n$.
\end{thm}
\begin{proof}
Recall again from Section~\ref{sec:prelim} that both the wave operator
$\square$ and its causal Green function $\G$ commute with $\d$ and hence
constitute cochain maps between the de~Rham complexes with appropriate
supports, inducing maps in cohomology. Moreover, since $\square$ is
cochain homotopic to zero [Equation~\eqref{eq:dal-def}], it induces the
zero map in cohomology.

Let us start with spacelike compact supports. We can break the exact
sequence in~\eqref{eq:short-sc} into two short exact sequences of
complexes:
\begin{equation}
\begin{tikzcd}
	0 \arrow{r} &
	\Omega_c^p(M) \arrow{r}{\square} &
	\Omega_c^p(M) \arrow{r}{\G} &
	\Omega_{\square,\SC}^p(M) \arrow{r} &
	0 ,
\end{tikzcd}
\end{equation}
\begin{equation}
\begin{tikzcd}
	0 \arrow{r} &
	\Omega_{\square,\SC}^p(M) \arrow{r}{\sso} &
	\Omega_{\SC}^p(M) \arrow{r}{\square} &
	\Omega_{\SC}^p(M) \arrow{r} &
	0 .
\end{tikzcd}
\end{equation}
Because $\square$ always induces the zero map, $[\square] = 0$, the
corresponding long exact sequences in cohomology
[cf.~Equation~\eqref{eq:long-ex}] break up into the following short
exact sequences:
\begin{equation}
\begin{tikzcd}
	0 \arrow{r} &
	H^p_c(M) \arrow{r}{[\G]} &
	H^p_{\square,\SC}(M) \arrow{r}{[\d]} &
	H^{p+1}_c(M) \arrow{r} &
	0 ,
\end{tikzcd}
\end{equation}
\begin{equation}
\begin{tikzcd}
	0 \arrow{r} &
	H_{\SC}^{p-1}(M) \arrow{r}{[\d]} &
	H_{\square,\SC}^p(M) \arrow{r}{[{\sso}]} &
	H_{\SC}^p(M) \arrow{r} &
	0 ,
\end{tikzcd}
\end{equation}
again with the convention that any $H^p_X(M)$ vanishes for $p<0$ or
$p>n$. Since we are dealing with real vector spaces, any exact sequence
splits, giving us the isomorphisms
\begin{equation*}
	H^p_c(M) \oplus H^{p+1}_c(M)
	\cong H^p_{\square,\SC}(M) \cong
	H^{p-1}_{\SC}(M) \oplus H^p_{\SC}(M) .
\end{equation*}
Given that $H^0_c(M)$ and $H^{-1}_{\SC}(M)$ both vanish ($M$ is
non-compact and there are no forms in degree $p=-1$), plugging $p=0$
into the above isomorphism implies $H^0_{\SC}(M) \cong H^1_c(M)$.
Proceeding by induction on $p$, we can check that $H^p_{\SC}(M) \cong
H^{p+1}_c(M)$ for all $p$. Thus, we obtain the isomorophisms
\begin{align}
	H^p_{\SC}(M) &\cong H^{p+1}_c(M), \\
	H^p_{\square,\SC}(M) & 
			\cong H^p_c(M) \oplus H^{p+1}_c(M).
\end{align}
Applying the same argument to the exact sequence~\eqref{eq:short-tc}, we
obtain the isomorphisms
\begin{align}
	H^p_{\tc}(M) &\cong H^{p-1}(M), \\
	H^p_{\square}(M) &
			\cong H^p(M) \oplus H^{p-1}(M).
\end{align}
This completes the proof.
\end{proof}

Let $\Sigma\sso M$ be a Cauchy surface. Recall that, by the smooth
Geroch splitting theorem, we can always
smoothly factor $M \cong \R\times \Sigma$. This observation results in
\begin{cor}\label{cor:cauchy}
We have the isomorphisms
\begin{align}
	H^p_{\SC}(M) &\cong H^p_c(\Sigma), &
		H^p_{\square,\SC}(M) &\cong H^p_c(\Sigma) \oplus H^{p-1}_c(\Sigma) \\
	H^p_{\tc}(M) &\cong H^{p-1}(\Sigma), & \text{and} ~
	H^p_{\square}(M) &\cong H^p(\Sigma) \oplus H^{p-1}(\Sigma),
\end{align}
with the convention that all cohomologies vanish in degree $p$ for $p<0$
or $p>n$.
\end{cor}
\begin{proof}
The splitting $M \cong \R\times \Sigma$ shows that $M$ is homotopic to
$\Sigma$. Hence, by the homotopy invariance of de~Rham cohomologies with
unrestricted supports, we have the isomorphism $H^p(M) \cong
H^p(\Sigma)$. On the other hand, Poincar\'e duality induces the
isomorphism $H^p_c(M) \cong H^{p-1}_c(\Sigma)$. Therefore, the desired
conclusion follows directly from these identities in combination with
Theorem~\ref{thm:sc-tc}.
\end{proof}

Finally, knowing the respective de~Rham cohomologies with spacelike and
timelike compact supports, we have the following generalization of the
Poincar\'e lemma.
\begin{cor}\label{cor:pairing}
The non-degenerate bilinear pairing between $\Omega^p_{\SC}(M)$ and
$\Omega^{n-p}_{\tc}(M)$ descends to a non-degenerate bilinear pairing
between $H^p_{\SC}(M)$ and $H^{n-p}_{\tc}(M)$. There exists also a
non-degenerate bilinear pairing between $H^p_{\square,\SC}(M)$ and
$H^{n-p}_{\square}(M)$.
\end{cor}
\begin{proof}
A consequence of Theorem~\ref{thm:sc-tc} is that $H^p_{\SC}(M) \cong
H^{p+1}_c(M)$ and $H^{n-p}_{\tc}(M) = H^{n-p-1}(M)$. So, the usual
Poincar\'e duality establishes that $H^p_{\SC}(M)^* \cong
H^{n-p}_{\tc}(M)$.  The isomorphism can be exhibited by bilinear pairing,
which descends from the standard bilinear pairing between
$\Omega^p_{\SC}(M)$ and $\Omega^{n-p}_{\tc}(M)$, tracing its effect
throughout the proof of Theorem~\ref{thm:sc-tc}. Its non-degeneracy is
also a consequence of the Poincar\'e lemma applied to $H^p_c(M)$ and
$H^{n-p}(M)$.

It also follows from Theorem~\ref{thm:sc-tc} that $H^p_{\square,\SC}(M)
\cong H^p_c(M) \oplus H^{p+1}_c(M)$ and $H^{n-p}_{\square}(M) \cong
H^{n-p}(M) \oplus H^{n-p-1}(M)$. Again, the usual Poincar\'e duality
establishes the isomorphism $H^p_{\square,\SC}(M)^* \cong
H^{n-p}_{\square}(M)$. The isomorphism can be exhibited by a bilinear
pairing between $\Omega^p_{\square,\SC}(M)$ and
$\Omega^{n-p}_{\square}(M) \cong
\Omega^{n-p}_{\tc}(M)/\square\Omega^{n-p}_{\tc}(M)$, defined by the latter
identity and the self-adjointness of $\square$ with respect to our
pairing between forms. Again, tracing this pairing through the proof of
Theorem~\ref{thm:sc-tc} and appealing to the standard Poincar\'e duality
establishes its non-degeneracy.
\end{proof}

As already discussed in the introduction, the importance of knowing the
above cohomology groups is important for understanding the
(pre)symplectic and Poisson structure of classical field theories, as
emphasized in~\cite{kh-peierls,kh-big,sdh,bds,benini}. The same result
as Corollary~\ref{cor:cauchy} was obtained independently
in~\cite{benini}. As a matter of fact, the method of~\cite{benini} can
be seen as a special case of our homological calculation, as
discussed more explicitly at the end of Section~\ref{sec:conal-cohom}.

\section{Calabi or Killing-Riemann-Bianchi complex}\label{sec:calabi}
In~\cite{kh-peierls,kh-big}, it was pointed out that the construction of
the symplectic and Poisson structures on the phase space of field
theories with constraints and/or gauge invariance can be done using a
general framework, provided a given field theory satisfies certain
geometric conditions. These conditions include the existence of certain
differential complexes that extend the operators that constitute the
constraints and that generate the gauge transformations. For Maxwell
(and similar) theories, all of these complexes are invariably part of
the de~Rham complex~\cite[Secs.4.2--3]{kh-peierls}. On the other hand,
for linearized gravity, one has to use something different.
Unfortunately, the explicit form of these differential complexes is not
currently known for linearized gravity on an arbitrary
background~\cite[Sec.4.4]{kh-peierls}. However, in the special case of
constant curvature backgrounds, the answer is known and it is the
so-called \emph{Calabi complex}~\cite{calabi}. It is likely that, once
an explicit understanding of the corresponding differential complexes
for more general backgrounds is achieved, the general framework
of~\cite{kh-peierls,kh-big} would supersede recent covariant treatments
of the quantization of linearized gravity
like~\cite{fewster-hunt,hack-lingrav}.

The Calabi complex provides a fine resolution~\cite[Sec.II.9]{bredon} of
the sheaf of Killing vectors, similarly to how the de~Rham complex
provides a fine resolution of the sheaf of locally constant functions.
The cohomology of a sheaf (a rather abstract object) is isomorphic to
the cohomology of the complex of global sections of a fine resolution of
the same sheaf (a more concrete object), which is what makes fine
resolutions significant~\cite[Thm.II.4.1]{bredon}. As such, the Calabi
complex has been studied in some literature on the deformation of
constant curvature geometric
structures~\cite{calabi,bbl,goldschmidt-calabi,gasqui-goldschmidt-fr,gasqui-goldschmidt,pommaret,eastwood}.
Because its structure is substantially different from the de~Rham
complex, we summarize some of its relevant properties in
Sections~\ref{sec:calabi-tensor} through~\ref{sec:calabi-coh} before
concentrating on its causally restricted cohomologies in
Section~\ref{sec:calabi-caus}. Many of these properties are scattered
throughout or are simply not available in the existing literature. We
defer a fuller discussion of the Calabi complex, which collects these
properties and their proofs, to~\cite{kh-calabi}. However, all that we
really need for the purposes of Section~\ref{sec:calabi-caus} is the
existence of differential operators listed in
Section~\ref{sec:calabi-ops} and the identities between them. Since
these differential operators are explicitly given, the identities can
in principle be verified by direct calculation.

\subsection{Tensor bundles}
\label{sec:calabi-tensor}
We will present later a differential complex whose nodes are
sections of tensor bundles that are not so easy to express in
conventional notation. So, let us introduce the following short-hands.
We denote the cotangent bundle by $VM = T^*M$ and the bundle of metrics
(symmetric, covariant 2-tensors) by $S^2M = S^2T^*M$. Let $RM \sso
(T^*)^4M$ denote the sub-bundle of covariant 4-tensors that satisfy the
algebraic symmetries of the Riemann tensor ($R_{(ab)cd} = R_{ab(cd)} =
R_{abcd} - R_{cdab} = R_{[abc]d} = 0$). Next, we let $BM \sso (T^*)^5M$
denote the target bundle of the Bianchi operator $\nabla_{[a} R_{bc]de}$.
At this point it is convenient to notice that the fiber of each of these
bundles carries~\cite{fkwc} an irreducible representation of $\GL(n)$,
with $n = \dim M$. In fact, it is easiest to describe the remaining
tensor bundles in terms of the irreducible $\GL(n)$ representation
carried by their fibers. So let $C_lM \sso (T^*)^{l+2}M$ (with $C$
standing for \emph{Calabi}) denote the sub-bundles of covariant
$(l+2)$-tensors with the corresponding irreducible representations
listed in Table~\ref{tbl:bundles}, which also lists their fiber ranks.
It is consistent for us to assign $C_0M \cong VM$, $C_1M \cong S^2M$ and
$C_2M \cong RM$ and $C_3M \cong BM$. Recall that, on an $n$-dimensional
manifold, the largest rank of a fully antisymmetric tensor is $n$. So
the bundles $C_lM$ become trivial (zero fiber rank) for $l>n$.

\begin{table}
\caption{
It is conventional to label irreducible $\GL(n)$ representations by
\emph{Young diagrams}~\cite{fulton}. Recall that a Young diagram with
$k$ cells of type $(r_1,r_2,\ldots)$ consists of a number of rows of
non-increasing lengths $r_i$, $r_{i+1} \le r_i$, such that $\sum_i r_i =
k$. Given a Young diagram with $k$ cells, an instance of the
corresponding irreducible $\GL(n)$ representation class can be realized
as the image of the space of covariant $k$-tensors after two
projections: assign an independent tensor index to each cell of the
diagram, symmetrize over each row, antisymmetrize over each
column.}
\hspace{1em}%
The table below lists the tensor bundles of the Calabi complex, the
corresponding irreducible $\GL(n)$ representations (labeled by Young
diagrams), and their fiber ranks, for $\dim M = n$. The rank is given by
the famous \emph{hook formula}, which is the following fraction. The
numerator is the product of the following numbers: place $n$ in the top
left cell, increase by $1$ to the right and decrease by $1$ down, until
all cells are filled. The denominator is the product of the following
numbers: fill a given cell with the number of cells constituting a hook
with vertex at the given location, extending to the right and
down~\cite{fulton}.
\label{tbl:bundles}
\begin{center}
\begin{tabular}{ccc}
	bundle & Young diagram & fiber rank \\
	\hline \rule[0.5ex]{0pt}{2.5ex} \\[-3ex]
	$VM\cong C_0M$
		& \yd{1} & $n$
		\\ \\[-1ex]
	$S^2M \cong C_1M$
		& \yd{2} & $\frac{n(n+1)}{2}$ \\ \\[-1ex]
	$RM\cong C_2M$
		& \yd{2,2} & $\frac{n^2(n^2-1)}{12}$ \\ \\
	$BM \cong C_3M$
		& \yd{2,2,1} & $\frac{n^2(n^2-1)(n-2)}{24}$ \\ \\[-2ex]
	\hline \rule[0.5ex]{0pt}{2.5ex} \\[-3ex]
	$C_lM$
		& \yt{{1}{},{2}{},{\none[\raisebox{-.5ex}{\vdots}]},{l}}
		& $\frac{n^2(n^2-1)(n-2)\cdots(n-l+1)}{2(l+1)l(l-2)!}$ \\ \\[-2ex]
	\hline \rule[0.5ex]{0pt}{2.5ex}
\end{tabular}
\end{center}
\end{table}

Given two $S^2M$ tensors, we can construct an $RM$ tensor out of them
using the formula
\begin{equation}
	(g \odot h)_{abcd}
	= g_{ac} h_{bd} - g_{bc} h_{ad} - g_{ad} h_{bc} + g_{bd} h_{ac} .
\end{equation}
In fact, the above formula represents a $\GL(n)$-equivariant map between
$S^2\otimes S^2$ and $R$ (where we use the bundle prefixes to stand in
for the corresponding irreducible representations). The decomposition of
the $S^2\otimes S^2$ tensor product has only one copy of $R$, so by
Schur's lemma such a map is unique, up to an overall rescaling. The same
argument can be repeated for the tensor product $S^2\otimes Y$, where
$Y$ corresponds to any other Young diagram. This tensor product
decomposes into irreducible subrepresentations without multiplicities.
Then the projection onto any of the subrepresentations $Y'$ is
well defined up to a rescaling. If we fix sections $g$ of $S^2M$ and
$h$ of $YM$, these projections define a bilinear operation between $g$
and $h$ with the result a section of $Y'M$. We use the following
explicit formulas:
\begin{align}
\notag
	(g\odot t)_{abc:de}
		&=
			+ g_{ad} t_{bc:e} + g_{bd} t_{ca:e} + g_{cd} t_{ab:e} \\
		&\quad {}
			- g_{ae} t_{bc:d} - g_{be} t_{ca:d} - g_{ce} t_{ab:d} , \\
\notag
	(g\odot t)_{abcd:ef}
		&=
			+g_{ae}t_{bcd:f}-g_{be}t_{cda:f}+g_{ce}t_{dab:f}-g_{de}t_{abc:f} \\
		&\quad {}
			-g_{af}t_{bcd:e}+g_{bf}t_{cda:e}-g_{cf}t_{dab:e}+g_{df}t_{abc:e} .
\end{align}
Note that a tensor with indices written as in $t_{abc:de}$ has the
symmetry type $(2,2,1)$, while $t_{abc:d}$ corresponds to the symmetry
type $(2,1,1)$, and so on. The colon $:$ is used purely as a visuall aid
to separate groups of indices belonging to different columsn of a Young
diagram.

The metric $g_{ab}$ itself, an $S^2M$ tensor, can now be used to produce
an $RM$ tensor,
\begin{equation}
	(g\odot g)_{ab:cd} = 2(g_{ac}g_{bd} - g_{bc}g_{ad}) ,
\end{equation}
which is obviously covariantly constant. In fact, a constant curvature
spacetime must have (covariant) Riemann tensor, Ricci tensor and Ricci
scalar of the following form
\begin{equation}\label{eq:bkg-riem}
	\bar{R}_{abcd} = \frac{k}{n(n-1)} (g_{ac}g_{bd} - g_{bc}g_{ad}), \quad
	\bar{R}_{ac} = \frac{k}{n} g_{ac} , \quad
	\bar{R} = k .
\end{equation}
We have decorated these quantities with a bar to indicate the fact that
we shall fix a constant curvature background metric $g$ and consider
perturbations on it. For our purposes, we also require that the
Lorentzian manifold $(M,g)$ is globally hyperbolic.

We should note that solutions of Einstein equations (including a
possible cosmological constant term) with constant curvature includes
Minkowski space ($k=0$), de~Sitter space ($k>0$) and anti-de~Sitter
space ($k<0$). There is (up to isometry) a unique simply connected
version of each of these cases~\cite[\S 5.1--2]{hawking-ellis}.
Other examples may be obtained by taking
quotients thereof with respect to a discrete subgroup, thus changing the
topology. The list of possibilities is thus exhausted by considering
open subsets of such quotients. Some examples will not be globally
hyperbolic (like anti-de~Sitter space or quotients of Minkowski space
with respect to timelike translations) and thus excluded from part of
our discussion.

\subsection{Differential operators}
\label{sec:calabi-ops}
Now, we introduce a number of differential operators between the tensor
bundles that we have defined. For convenience of notation, we denote the
space of sections of a bundle by the same symbol as the bundle itself.
These operators fit into the following diagram:
\begin{equation}
\begin{tikzcd}
	0 \arrow{r} &
	C_0M \arrow{r}{B_1} \arrow{d}{P_{0}} &
	C_1M\arrow{r}{B_2} \arrow{d}{P_{1}} \arrow[dashed]{dl}[swap]{E_1} &
	C_2M \cdots \arrow{r}{B_{n}} \arrow{d}{P_2} \arrow[dashed]{dl}[swap]{E_2} &
	C_{n}M \arrow{r} \arrow{d}{P_{n}}
	\arrow[dashed]{dl}[swap]{E_{n}} &
	0 \\
	0 \arrow{r} &
	C_0M \arrow{r}{B_1} &
	C_1M\arrow{r}{B_2} &
	C_2M \cdots \arrow{r}{B_{n}} &
	C_{n}M \arrow{r} &
	0 \\
\end{tikzcd} .
\end{equation}
All the solid arrows commute and the rows constitute (cochain)
complexes. The vertical maps are then necessarily cochain maps. They
happen to satisfy the identities $P_l = E_{l+1}\circ B_{l+1} + B_l\circ
E_l$, which means that they are null-homotopic, with the $E_l$ supplying
the corresponding cochain homotopy.

Below, we give explicit formulas for all these differential operators in
dimension $n=4$. More details can be found in~\cite{kh-calabi}, which
draws from the earlier
works~\cite{calabi,bbl,goldschmidt-calabi,gasqui-goldschmidt-fr,gasqui-goldschmidt,pommaret,eastwood}.
As we shall see, for low indices they are well known in the relativity
literature. However, the relations between them in terms of fitting into
the above diagram do not seem to have been fully noted.

The \emph{Calabi differential complex} is given by
\begin{align}
	B_1[v]_{ab} &= \nabla_a v_b + \nabla_b v_a , \\
\notag
	B_2[h]_{ab:cd} &= \left(
		\nabla_{(a}\nabla_{c)} h_{bd} -
		\nabla_{(b}\nabla_{c)} h_{ad} -
		\nabla_{(a}\nabla_{d)} h_{bc} +
		\nabla_{(b}\nabla_{d)} h_{ac} \right) \\
		& \qquad {}
			+ k\frac{1}{n(n-1)} (g\odot h)_{ab:cd} , \\
	B_3[r]_{abc:de} &= 3\nabla_{[a} r_{bc]:de}
		= \nabla_a r_{bc:de} + \nabla_b r_{ca:de} + \nabla_c r_{ab:de} , \\
	B_4[b]_{abcd:ef}
		&= 4\nabla_{[a} b_{bcd]:ef} \\
		&= \nabla_a b_{bcd:ef} - \nabla_b b_{cda:ef}
			- \nabla_c b_{dab:ef} - \nabla_d b_{abc:ef} , \\
	B_l[b]_{a_1\cdots a_l:bc}
		&= l\nabla_{[a_1} b_{a_2\cdots a_l]:bc} \quad (l \ge 3) ,
\end{align}
where $(a_1 \cdots a_l)$ and $[a_1 \cdots a_l]$ denote respectively
complete idempotent symmetrization and antisymmetrization of a group of
indices~\cite[Eq.2.4.3--4]{wald}. Recall also that the colon $:$ is used
purely as a visual aid to separate groups of indices belonging to
different columns of the Young diagrams in Table~\ref{tbl:bundles}.
The details showing that these operators have the desired symmetry
properties and indeed define a complex, $B_{l+1}\circ B_l = 0$, which is
moreover elliptic,%
	\footnote{A complex of differential operators is \emph{elliptic} if
	the corresponding complex of symbol maps is exact for every non-zero
	covector.} %
can be found in~\cite{kh-calabi}.

It is interesting to note the following relations with well known
differential operators in relativity. The \emph{Killing} operator is
$K[h] = B_1[h]$. The \emph{linearized Riemann} tensor is $\dot{R}[h] =
-\frac{1}{2} B_2[h] + k \frac{2}{n(n-1)} (g\odot h)$, where the all
covariant non-linear Riemann tensor is expanded as $R[g+\lambda
h]_{ab:cd} = \bar{R}_{ab:cd} + \lambda \dot{R}[h]_{ab:cd}$ (convention
of~\cite{wald}). The background \emph{Bianchi} operator is $\bar{B}[r] =
B_3[r]$, with $\bar{B}[\bar{R}] = 0$. Finally, though the name is not
standard, it is meaningful to call $B_4[b]$ a \emph{higher Bianchi}
operator. Thus, it would also make sense to refer to the Calabi complex
as the \emph{Killing-Riemann-Bianchi complex}. This complex also happens
to be locally exact%
	\footnote{A differential complex on a manifold $M$ is \emph{locally
	exact} if every $x\in M$ has a neighborhood such that the complex
	restricted to it becomes exact. For example, this condition is
	fulfilled for the de~Rham complex thanks to the Poincar\'e lemma.} %
~\cite{calabi,kh-calabi}. Thus, according to the general machinery of
sheaf theory, the Calabi complex provides a fine resolution of the
\emph{sheaf of Killing vectors} (or \emph{Killing sheaf}) $\K_g$ on the
Lorentzian manifold $(M,g)$~\cite[Sec.3]{kh-calabi}. This observation
immediately gains us the following
\begin{prop}[Calabi \cite{calabi}]
The (unrestricted) cohomology $HC^l(M,g) = \ker B_{l+1}/\im B_l$ of the
Calabi complex is isomorphic to the sheaf cohomology $H^\bullet(M,\K_g)$
of the sheaf of Killing vectors on any spacetime $(M,g)$ of constant
curvature.
\end{prop}
Calabi's proof was rather elementary and relied on the specific
structure of this complex. Unfortunately, his method does not generalize
easily to other differential complexes. So, we discuss below a different
method to get local exactness, which relies mostly on the ellipticity of
the Calabi complex, a property which is expected to be shared by other
complexes of interest.

Next, we give explicitly the homotopy differential operators
\begin{align}
	E_1[h]_a
		&= D[h]_a = \nabla^b h_{ab} - \frac{1}{2} \nabla_a h , \\
	E_2[r]_{a:b}
		&= \tr[r]_{a:b} = r_{ac:b}{}^c , \\
\notag
	E_3[b]_{ab:cd}
		&= \nabla^e b_{e ab:cd}
			+ \frac{1}{2}\nabla^e (b_{c ab:d e} - b_{d ab:c e}) \\
\notag
		& \quad {}
			-\frac{1}{2} (\nabla_c b_{ab e:d}{}^e - \nabla_d b_{ab e:c}{}^e) \\
\notag
		& \quad {} - \frac{1}{2}
			(\nabla_a b_{c b e:d}{}^e - \nabla_a b_{d b e:c}{}^e \\
		& \qquad {}
			+\nabla_b b_{a c e:d}{}^e - \nabla_b b_{a d e:c}{}^e) , \\
\notag
	E_4[b]_{abc:de}
		&= \nabla^f b_{f abc:de}
			+ \frac{1}{3}\nabla^f (b_{d abc:e f} - b_{e abc:d f}) \\
\notag
		& \quad {}
			+ \frac{1}{3} (\nabla_d b_{abc f:e}{}^f - \nabla_e b_{abc f:d}{}^f) \\
\notag
		& \quad {} + \frac{1}{6}
			(\nabla_a b_{d bc f:e}{}^f - \nabla_a b_{e bc f:d}{}^f \\
\notag
		& \qquad {}
			+\nabla_b b_{a d c f:e}{}^f - \nabla_b b_{a e c f:d}{}^f \\
		& \qquad {}
			+\nabla_c b_{ab d f:e}{}^f - \nabla_c b_{ab e f:d}{}^f) , \\
\notag
	E_{l+1}[b]_{a_1\cdots a_l:bc}
		&= \nabla^a b_{a a_1\cdots a_l:bc}
			+ l^{-1}\nabla^a (b_{b a_1\cdots a_l:c a} - b_{c a_1\cdots a_l:b a}) \\
\notag
		& \quad {}
			-\frac{(-1)^l}{l}
				(\nabla_b b_{a_1\cdots a_l a:c}{}^a
				-\nabla_c b_{a_1\cdots a_l a:b}{}^a) \\
\label{eq:calabi-homot-end}
		& \quad {}
			-\frac{(-1)^l}{l(l-1)}
				(\nabla_{\{b\}} b_{\{a_1\cdots a_l\}a:c}{}^a
				-\nabla_{\{c\}} b_{\{a_1\cdots a_l\}a:b}{}^a) \\
\notag
		& \quad {}
			\qquad \text{for $(l \ge 2)$, where} \\
\notag
	p_{\{b\}} t_{\{a_1\cdots a_l\}}
		&= \sum_{i=1}^l (-1)^{i+1} p_{a_i} t_{a a_1\cdots \hat{a}_i \cdots a_l}
			\quad \text{($\hat{a}_i$ omitted).}
\end{align}
Their desired Young symmetry properties are demonstrated
in~\cite{kh-calabi}.  Again, we find the following relations with
classical differential operators from relativity. The \emph{de~Donder}
operator is $D[h] = E_1[h]$. The trace from the Riemann to the Ricci
tensors is given by $\bar{R}_{ab} = \bar{R}_{ac:b}{}^{c} =
E_2[\bar{R}]_{ab}$. The higher homotopy operators $E_l$ do not seem to
be part of the classical literature. However, they are essentially
modified divergence operators and are thus reminiscent of the de~Rham
co-differentials.

Finally, the cochain maps $P_l = E_{l+1}\circ B_{l+1} + B_l\circ E_l$
(with the edge cases $P_0 = E_1\circ B_1$ and $P_n = B_n \circ E_n$) are
given by
\begin{align}
	P_0[v]_a
		&= \square v_a + k\frac{1}{n} v_a , \\
	P_1[h]_{ab}
		&= \square h_{ab} 
			- k \frac{2}{n(n-1)} h_{ab}
			+ 2k \frac{g_{ab} \tr[h]}{n(n-1)}  , \\
	P_2[r]_{ab:cd} 
		&= \square r_{ab:cd}
			- k\frac{2}{n} r_{ab:cd}
			+ 2k\frac{(g\odot \tr[r])_{ab:cd}}{n(n-1)}  , \\
	P_3[b]_{abc:de}
		&= \square b_{abc:de}
			- k\frac{(3n-7)}{n(n-1)} b_{abc:de}
			- 2k\frac{(g\odot \tr[b])_{abc:de}}{n(n-1)}  , \\
	P_4[b]_{abcd:ef}
		&= \square b_{abcd:ef}
			- k\frac{(4n-14)}{n(n-1)} b_{abcd:ef}
			+ 2k\frac{(g\odot \tr[b])_{abcd:ef}}{n(n-1)}  , \\
\notag
	P_l[b]_{a_1\cdots a_l:bc}
		&= \square b_{a_1\cdots a_l:bc}
			- k\frac{(ln-l^2+2)}{n(n-1)} b_{a_1\cdots a_l:bc} \\
		& \qquad {}
			+ (-)^l 2k \frac{(g\odot \tr[b])_{a_1\cdots a_l:bc}}{n(n-1)}
			\quad (l\ge 3) .
\end{align}
where we have defined the traces as $\tr[h] = h_e{}^e$,
$\tr[r]_{ab} = r_{a e: b}{}^e$, $\tr[b]_{ab:c} = b_{ab e: c}{}^e$,
$\tr[b]_{abc:d} = b_{abc e: d}{}^e$, and $\tr[b]_{a_1\cdots a_l:b} =
b_{a_1\cdots a_l a:b}{}^a$. The required null-homotopy identities $P_l =
E_{l+1}\circ B_{l+1} + B_l\circ E_l$ (including the edge cases $P_0 =
E_1\circ B_1$ and $P_n = B_n \circ E_n$) are demonstrated
in~\cite{kh-calabi}. These identities for $P_0[v]$ and $P_1[v]$ are well
known and are tightly linked with the de~Donder gauge fixing condition
in linearized gravity~\cite{wald,fewster-hunt}. The higher cochain maps
and the corresponding identities appear to be new.  Though, the identity
for $P_2[r]$ is related to the non-linear wave equations satisfied by
the Riemann and Weyl tensors on any vacuum background, sometimes known
as the \emph{Lichnerowicz Laplacian}~\cite[Sec.1.3]{lichnerowicz} (see
also~\cite[Sec.7.1]{chr-kl}, \cite[Exr.15.2]{mtw}, \cite[Eq.35]{bcjr}).

\subsection{Cohomology with unrestricted and compact supports}
\label{sec:calabi-coh}
Let us denote the cohomology of the Calabi complex by $HC^l_X(M,g)$,
where $X=c,+,-,\fc,\pc,\SC,\tc$ or empty, according to the conventions of
Section~\ref{sec:prelim}. As in the case of the de~Rham complex in
Section~\ref{sec:comp}, we will later relate the cohomology with
causally restricted supports to that with unrestricted or compact
supports. It remains still to find a means to calculate these cohomology
groups. We will state some results in that direction below, referring
to~\cite{kh-calabi} for a fuller discussion.

An important observation is that each of the $P_l$ operators is
wave-like, that is, it has the same principal symbol as the wave
operator $\square_g$ with respect to the background Lorentzian metric
$g$. This observation has a dual role. First, this means that each of
the $P_l$ operators is Green hyperbolic~\cite{bgp,baer}, while being
cochain homotopic to zero, opening the door to using the methods of
Section~\ref{sec:comp} to compute the cohomology with causally
restricted supports.

The second role is more subtle:
\begin{rem}\label{rem:elliptic}
Note that the principal symbols of the $B_l$ maps in the Calabi complex
are actually $\GL(n)$-equivariant and so do not actually involve the
background metric $g$. On the other hand, the principal symbols of the
cochain maps $P_l$ do depend on $g$. This dependence comes purely from
the cochain homotopy operators $E_l = E^g_l$ and the identity $P_l =
P_l^g = E^g_{l+1}\circ B_{l+1} + B_l\circ E_l^g$, where we have used the
subscript $g$ to indicate that the background metric was used for
covariant differentiation and index raising. On the other hand, we are
completely free to define a different set of cochain maps $P_l^{g_R} =
E^{g_R}_{l+1}\circ B_{l+1} + B_l\circ E_l^{g_R}$, which now depend on a
different metric $g_R$ with Riemannian signature. It is crucial to note
that the principal symbol of $P^{g_R}_l$ depends only on the principal
symbols of the $E_l^{g_R}$ and $B_l$. So, in fact, it is equal to the
principal symbol of $P^g_l$, but with the Lorentzian metric $g$ replaced
by the Riemannian metric $g_R$.  In other words, each of the $P_l^{g_R}$
operators is elliptic, since its principal symbol coincides with the
Laplace operator $\Delta_{g_R}$. Of course, $P_l^{g_R}$ would differ
much more radically from the formulas we have given for $P_l^g$ in the
terms of subleading differential orders.
\end{rem}

The ellipticity of the complex (together with a subtler property known
as a $\delta$-estimate, discussed in more detail
in~\cite{tarkhanov,kh-calabi}) results in the following
\begin{prop}
Let us denote by $\Secs(C_lM)$ the space of smooth sections of the
tensor bundle $C_lM\to M$.
(a) The cohomology $HC^\bullet(M,g)$ of the Calabi complex
$(\Secs(C_lM),B_l)$ is isomorphic to the cohomology $H^\bullet(M,\K_g)$
of the sheaf $\K_g$ of Killing vectors on $(M,g)$. (b) If $(M,g)$ is a
simply connected, constant curvature Lorentzian manifold, then
$H^\bullet(M,\K_g) \cong H^\bullet(M)\otimes V_g$, where $V_g$ is the
vector space of all Killing vectors and $H^\bullet(M)$ is the de~Rham
cohomology group.
\end{prop}
Killing vectors (or rather covectors in our notation) are solutions
$v\in \Secs(T^*M)$ of the Killing equation $K[v]_{ab} = \nabla_a v_b +
\nabla_b v_a = 0$. On simply connected, constant curvature
$n$-dimensional spacetimes, $\dim V_g = {n+1\choose 2}$. Note also that
the simple connectedness condition implies that $H^1(M) = 0$. The
precise definition of a sheaf and its cohomology is not of particular
importance for the moment. For present purposes, it suffices that the
above result, at the very least, answers the question of what
$HC^\bullet(M,g)$ is for the simply connected versions of Minkowski
($\R^n$), de~Sitter ($\R\times S^{n-1}$ for $n\ge 3$, $\R^2$ for $n=2$)
and anti-de~Sitter ($\R^n$) spacetimes. The proof, together with a
partial discussion of the non-simply connected case, can be found
in~\cite{kh-calabi}.

It remains to discuss Calabi cohomology with compact supports
$HC_c^\bullet(M,g)$. First, we note that the chain complex
$(\Secs(C^*_lM),B^*_l)$ formally adjoint to the Calabi complex has the
interesting property that equation $B^*_n[b] = 0$ is equivalent to the
\emph{(rank-$(n-2)$) Killing-Yano} equation $Y[w]_{abc_4\cdots c_n} =
\nabla_{(a} w_{b)c_4\cdots c_n}$, where a solution with $w_{[bc_4\cdots
c_n]} = w_{bc_4\cdots c_n}$ is called a \emph{(rank-$(n-2)$)
Killing-Yano} tensor on $(M,g)$. We define \emph{Calabi homology}
$HC_l(M,g)$ as the cohomology of this adjoint complex
$(\Secs_c(C_l^*M),B_l^*)$ with compact supports and also \emph{locally
finite Calabi homology} as the cohomology of the adjoint complex
$(\Secs(C_l^*M),B^*_l)$ with unrestricte supports. Since taking formal
adjoints preserves the homotopy identities and ellipticity, appealing to
the same arguments as above (again, including a
$\delta$-estimate~\cite{tarkhanov,kh-calabi}) we also have
\begin{prop}
(a) Locally finite Calabi homology $HC^{\lf}_l(M,g)$ is isomorphic to the
cohomology $H^\bullet(M,\KY_g)$ of the sheaf $\KY_g$ of Killing-Yano
tensors on $(M,g)$. (b) If $(M,g)$ is a simply connected, constant
curvature Lorentzian manifold, then $H^\bullet(M,\KY_g) \cong
H^\bullet(M)\otimes W_g$, where $W_g$ is the vector space of all
Killing-Yano tensors and $H^\bullet(M)$ is the de~Rham cohomology group.
\end{prop}
On simply connected, constant curvature $n$-dimensional spacetimes,
$\dim W_g = {n+1\choose 2}$~\cite{stepanov}. Furthermore, using
Remark~\ref{rem:elliptic} and some general results from the theory of
elliptic differential complexes (see Example~5.1.11 of~\cite{tarkhanov},
which relies on the results of~\cite{serre}), we have the following
generalized Poincar\'e duality isomorphisms~\cite{kh-calabi}:
\begin{prop}
When finite dimensional, Calabi homology is the linear dual of Calabi
cohomology, $HC_l(M,g) = HC^l(M,g)^*$, while Calabi cohomology with
compact supports is the linear dual of locally finite Calabi
homology, $HC_c^l(M,g) = HC^{\lf}_l(M,g)^*$. In both cases, the duality
can be exhibited via the non-degeneracy of the pairing descended from
the natural pairing between the chains and cochains of corresponding
complexes.
\end{prop}

\subsection{Cohomology with causally restricted supports}
\label{sec:calabi-caus}
Recall that, in Section~\ref{sec:prelim}, we defined de~Rham
cohomologies $H^p_X(M)$ with causally restricted supports $X = +, -,
\SC, \tc, \pc$ or $\fc$ by restricting the de~Rham complex to forms with
supports indicated by $X$, with the on-shell cohomologies
$H^p_\square(M)$ and $H^p_{\square,\SC}(M)$. Substituting the Calabi
complex for the de~Rham complex and the $P_l$ operators for the
d'Alembertians $\square$, by direct analogy we can define the causally
restricted Calabi cohomologies $HC^l_X(M,g)$, $HC^l_P(M,g)$ and
$HC^l_{P,\SC}(M,g)$. We can use the same definitions also in the case of
the adjoint Calabi complex, with slightly altered notation. Let the
causally restricted Calabi homology $HC^X_l(M,g)$ be the cohomology of
the complex $(\Secs_Y(C_l^*M),B_l^*)$ where the pair $(X,Y)$ is one of
\emph{retarded} $(+,\fc)$, \emph{advanced} $(-,\pc)$ , \emph{spacelike
locally finite} $(\slf,\tc)$, \emph{timelike locally finite}
$(\tlf,\SC)$ , \emph{future locally finite} $(\flf,-)$  and \emph{past
locally finite} $(\plf,+)$. Similarly, we define the on-shell Calabi
homologies $HC^{P,\lf}_l(M,g)$ and $HC^{P,\tlf}(M,g)$ as the cohomologies of
the complexes $(\ker P^*_l\cap \Secs(C_l^*M), B_l^*)$ and $(\ker
P^*_l\cap \Secs_{\SC}(C_l^*M), B_l^*)$, respectively. The above $(X,Y)$
pairs are chosen specifically so that there is a bilinear pairing
between Calabi homology $HC^X_l(M,g)$ and Calabi cohomology
$HC_Y^l(M,g)$, which descends from the natural pairing between the
corresponding spaces of sections of $C^*_lM$ and $C_lM$.

The specific way in which these causally restricted cohomologies are of
importance in linearized gravity is summarized in the following
proposition. For definiteness of notation let us fix a $\chi \in
C^\oo(M)$ that is $1$ in the future of a Cauchy surface $\Sigma_+$ and
$0$ in the past of another Cauchy surface $\Sigma_-$. The following is a
special case of the general result~\cite[Thm.3.2]{kh-peierls}.
\begin{prop}\label{prp:lingrav}
Linearized gravity on a constant curvature
background~\cite[Sec.4.4]{kh-peierls} induces a symplectic form on
$\Secs_{P,\SC}(C_1M)$~\cite[Def.3.10]{kh-peierls} that is non-degenerate
when (a) the bilinear form on $HC^1_{\SC}(M,g) \times HC_1(M,g)$ induced
by $\langle \alpha, \beta \rangle = \int_M \alpha \cdot \beta$ is
non-degenerate and (b) the bilinear form on $HC^1_{P,\SC}(M)$ induced
by $\langle \alpha, \beta \rangle_P = \int_M \alpha \cdot P_1 [\chi
\beta]$ is non-degenerate.
\end{prop}
From the proof of that proposition it also follows that degeneracies in
(a) and (b) can imply degeneracies in the corresponding (pre-)symplectic
structure.

With the above discussion in mind, we can see immediately that we are in
a situation very similar to that of Section~\ref{sec:comp}, with the
de~Rham complex replaced by the Calabi complex (or its adjoint complex)
and the wave operators $\square$ replaced by the operators $P_l$ (or
$P^*_l$), which have wave-like principal symbols and are Green
hyperbolic. So, repeating the arguments of Section~\ref{sec:comp}, we
immediately have the following
\begin{thm}\label{thm:calabi}
Consider a globally hyperbolic, constant curvature Lorentzian manifold
$(M,g)$. The Calabi cohomology $HC^l_X(M,g)$ with the causally restricted
supports $X=+,-,\pc$ or $\fc$ is trivial. Moreover, for the cases
$X=\SC,\tc$, we have the isomorphisms
\begin{align}
	HC^l_{\SC}(M,g) &{\cong} HC^{l+1}_c(M,g), &
		HC^l_{P,\SC}(M,g) &{\cong} HC^l_c(M,g) {\oplus} HC^{l+1}_c(M,g), \\
	HC^l_{\tc}(M,g) &{\cong} HC^{l-1}(M,g), &
		HC^l_{P}(M,g) &{\cong} HC^l(M,g) {\oplus} HC^{l-1}(M,g),
\end{align}
with the convention that all cohomologies vanish in degree $l$ for $l<0$
or $l>n$. Similarly, the Calabi homology $HC_l^X(M,g)$ with the causally
restricted supports $X=+,-,\plf$ or $\flf$ is trivial. Moreover, for the
cases $X=\tlf,\slf$, we have the isomorphisms
\begin{align}
	HC_l^{\tlf}(M,g) &{\cong} HC_{l-1}(M,g), &
		\hspace{5.5em}%
		\llap{$HC^l_{P,\tlf}(M,g)$} &{\cong} HC_l(M,g) {\oplus} HC_{l-1}(M,g), \\
	HC_l^{\slf}(M,g) &{\cong} HC^{\lf}_{l+1}(M,g), &
		\llap{$HC^l_{P,\lf}(M,g)$} &{\cong} HC^{\lf}_l(M,g) {\oplus} HC^{\lf}_{l+1}(M,g),
\end{align}
again with the convention that all cohomologies vanish in degree $l$ for
$l<0$ or $l>n$.
\end{thm}
The Calabi cohomology with spacelike compact support in degree $l=1$ is
important in understanding the symplectic and Poisson structure of the
classical field theory (and of course the quantization) of linearized
gravitons on a background of constant curvature. This was pointed out
explicitly in~\cite[Sec.4.4]{kh-peierls} as a special case of a more
general phenomenon (also discussed in~\cite{kh-big}).

\begin{rem}
Using the above theorem and the results of Section~\ref{sec:calabi-coh},
we can assert that for $n$-dimensional Minkowski space $HC^l_{\SC}$
vanishes in all degrees except $l=n-1$, while $HC^l_{P,\SC}$ vanishes in
all degrees except $l=n,n-1$. For $n$-dimensional de~Sitter space
$HC^l_{\SC}$ vanishes in all degrees except $l=n-1$, while $HC^l_{P,\SC}$
vanish in all degrees except $l=0,n-1,n$. Similar remarks apply to
Calabi homologies.
\end{rem}

\section{Notes and generalizations}\label{sec:notes}

\subsection{Generalized causal structures}\label{sec:gen-caus}
The notion of a causal structure on a manifold or even a topological
space (in the sense of a partial order on events) can be generalized
quite fare beyond the context of Lorentzian geometry~\cite{kp,gps}. We
will stick with the context of differential geometry, where a natural
generalization consists of introducing at every point of a manifold an
arbitrary convex cone in the tangent%
	\footnote{One could equally do so in the cotangent bundle, and
	produce a tangent cone by convex (or polar) duality.} %
bundle. Such a manifold could be called a \emph{conal
manifold}~\cite{neeb,lawson,sullivan,kh-big}. Various notions generated
by the causal structure on Lorentzian manifolds survive almost without
modification on conal manifolds, including spacelike and timelike
compactness. The main question we will try to answer in this section is
the following: Is it possible to use the methods of
Section~\ref{sec:comp} to compute causally restricted cohomologies on a
conal manifold? We shall see that the answer is \emph{yes}, even if the
conal manifold is not Lorentzian.

\subsubsection{Conal manifolds}\label{sec:conal}
Before dealing with spacelike and timelike compactly supported forms,
let us introduce the basics of conal manifolds and causal relations on
them. Let $M$ be a smooth manifold and $C\subset TM$ be an open subset,
such that $C_x = C \cap T_x M$ is an open, convex cone in $T_x M$ that
does not contain any affine line. It can be shown that the interior
$C^\oast_x$ of the polar dual (or convex dual) cone $T^*_x M \supset C^*_x
= \{ p \in T^*_xM \mid \forall v\in C_x \colon p\cdot v \ge 0 \}$
satisfies the same conditions, with $C^\oast = \sqcup_{x\in M}
C^\oast_x$. The pair $(M,C)$ or $(M,C^\oast)$ is called a \emph{conal
manifold}, with $C$ (or $C^\oast$) called the tangent (or cotangent)
\emph{cone distribution} or \emph{cone bundle}. For example, the subset
of non-vanishing, future-pointing, timelike vectors on a Lorentzian
manifold with a time orientation satisfies the above conditions. In
general, the cones $C_x$ need not even have elliptic cross sections,
thus not be associated to any Lorentzian metric. The cones of future
pointing timelike vectors of linear symmetric hyperbolic PDE systems
also satisfy the same properties~\cite[Sec.4.1]{kh-big}. Sometimes, it
is also convenient to admit degenerate cases where the cones are not
open or contain some affine lines, but some special care must be taken
in those situations.

Given a conal manifold $(M,C)$ we can define a \emph{chronological
order} relation on the points of $M$. Namely, $x \ll y$ if there exists
a smooth curve $\gamma\colon [0,1] \to M$, such that $\gamma(0) = x$,
$\gamma(1) = y$ and $\dot{\gamma}(t) \in C$ for all $t\in [0,1]$. It can
be shown that the chronological order relation $I^+\sso M\times M$ is
open and transitive. We can also define the \emph{reverse chronological
order}, $I^-$, and \emph{chronological influence}, $I = I^+\cup I^-$,
relations in the obvious way. We avoid defining the analog of the
\emph{causal} order relation usually denoted by $J^+$, simply because we
have not made any hypotheses about the regularity of the set of causal
vectors ($\overline{C}_x \sso T_xM$). Given any set $K\sse M$, we denote
by $I^\pm(K)$ the set of all points of $M$ that respectively
chronologically precede ore are preceded by the points of $K$. In
general, $I^\pm(K)$ is not closed, even if $K$ is. So, for convenience
we define $\overline{I}^\pm(K) = \overline{I^\pm(K)}$. We also use the
notation $I(K) = I^+(K) \cup I^-(K)$ and $\overline{I}(K) =
\overline{I}^+(K) \cup \overline{I}^-(K)$. Note that $\overline{I}^\pm
\sse M\times M$ need not be transitive as relations.

The definition of a Cauchy surface $\Sigma \subset M$ is the usual one,
every inextensible smooth curve with timelike tangents must intersect
$\Sigma$ exactly once. It has recently been shown that the smooth
version of the Geroch splitting theorem~\cite{geroch,bs1,bs2}
generalizes to conal manifolds~\cite{fs}. So, \emph{globally
hyperbolicity} can be simply characterized by the existence of a Cauchy
surface. Also, the results of~\cite{sanders} should also directly carry
over to conal manifolds. Finally, we define the notions of
\emph{advanced}, \emph{retarded}, \emph{spacelike compact},
\emph{timelike compact}, \emph{future compact} and \emph{past compact}
exactly in the same way as in Section~\ref{sec:prelim}, with the
exception that we use the relations $\bar{I}^\pm$ and $\bar{I}$ instead
of the relations $J^\pm$ and $J$.%
	\footnote{We are not concerned with possible minor inconsistencies
	this substitution introduces in the case of Lorentzian manifolds with
	ill-behaved causal structures. In any case, we shall only apply these
	notions for globally hyperbolic spacetimes, where these differences do
	not appear.}

\subsubsection{Cohomology with causally restricted supports}\label{sec:conal-cohom}
Let $M$ be a globally hyperbolic conal manifold and $g$ an auxiliary
globally hyperbolic Lorentzian metric that induces another conal
structure on $M$ that is ``slower'' than the original one. That is,
$\Omega^p_{\pm_g}(M) \sse \Omega^p_{\pm}(M)$, which also implies that
$\Omega^p_{\SC_g}(M) \sse \Omega^p_{\SC}(M)$, while
$\Omega^p_{\fc_g,\pc_g}(M) \supseteq \Omega^p_{\fc,\pc}(M)$, and hence
$\Omega^p_{\tc_g}(M) \supseteq \Omega^p_{\tc}(M)$. Any conal manifold
admits a nowhere vanishing vector field (contract each cone to a ray and
select a vector from it), which is moreover everywhere future directed.
So, the existence of such an auxiliary Lorentzian metric follows from
the same known, general arguments that show the existence of Lorentzian
metrics on manifolds of vanishing Euler characteristic (i.e., admitting
a nowhere vanishing vector field)~\cite{beem-ehrlich,oneill}. The
``slowness'' requirement is implemented by making sure that the
Lorentzian timelike cones closely hug the directions singled out by the
above everywhere timelike vector field.

Let $\G_\pm$ denote once again the advanced and retarded Green functions
of the wave operator $\square_g$ defined with respect to $g$. Then it is
easy to see that the Green functions are still well defined and
injective as maps $\G_\pm \colon \Omega^p_c(M) \to \Omega^p_\pm(M)$.
Appealing to the same logic as in the standard
proofs%
	\footnote{Pick an exaustion of $M$ by compact sets and adapt a
	sequence of smooth ``step functions'' to this exaustion. Precomposing
	$\G_\pm$ with multiplication by these step functions gives a sequence
	of operators which converges to an operator with the desired extended
	domain.}%
~\cite{bgp,ginoux,kh-peierls,kh-big,baer}, we can extend the Green
functions to bijective maps $\G_\pm\colon \Omega^p_\pm(M) \to
\Omega^p_\pm(M)$ and $\G_\pm\colon \Omega^p_{\fc,\pc}(M) \to
\Omega^p_{\fc,\pc}(M)$, from which it is straightforward to establish
exactness of the following sequences, with $\G = \G_+ - \G_-$:
\begin{equation}\label{eq:short-sc-conal}
\begin{tikzcd}
	0 \arrow{r} &
	\Omega_{0}^p(M) \arrow{r}{\square} &
	\Omega_{0}^p(M) \arrow{r}{\G} &
	\Omega_{\SC}^p(M) \arrow{r}{\square} &
	\Omega_{\SC}^p(M) \arrow{r} &
	0 ,
\end{tikzcd}
\end{equation}
\begin{equation}\label{eq:short-tc-conal}
\begin{tikzcd}
	0 \arrow{r} &
	\Omega_{\tc}^p(M) \arrow{r}{\square} &
	\Omega_{\tc}^p(M) \arrow{r}{\G} &
	\Omega^p(M) \arrow{r}{\square} &
	\Omega^p(M) \arrow{r} &
	0 ,
\end{tikzcd}
\end{equation}
where the supports are restricted by the given conal structure on $M$
and not by that induced by the auxiliary Lorentzian metric $g$. Note
that the proofs would make use of the hypothesis that the given conal
structure is globally hyperbolic, specifically in the construction of
explicit splitting maps that demonstrate
exactness~\cite[Lem.2.1]{kh-peierls}. Thus, repeating the arguments
Section~\ref{sec:comp}, we establish the following generalization of
Theorems~\ref{thm:pf} and~\ref{thm:sc-tc}.
\begin{thm}\label{thm:conal}
Consider a globally hyperbolic conal manifold $M$. Its de~Rham
cohomology $H^p_X(M)$ with causally restricted supports $X=+,-,\pc$ or
$\fc$ is trivial. Moreover, we have the isomorphisms
\begin{align}
	H^p_{\SC}(M) &\cong H^{p+1}_c(M), &
		H^p_{\square,\SC} &\cong H^p_c(M) \oplus H^{p+1}_c(M), \\
	H^p_{\tc}(M) &\cong H^{p-1}(M), & \text{and} ~
		H^p_{\square}(M) &\cong H^p(M) \oplus H^{p-1}(M),
\end{align}
with the convention that all cohomologies vanish in degree $p$ for $p<0$
or $p>n$.
\end{thm}

It should be clear from the preceding discussion that there is nothing
inherently special in our use of the d'Alembertian $\square_g$, when it
comes to the calculation of de~Rham cohomologies with causally
restricted supports on a globally hyperbolic conal manifold $M$. It is
merely one of multiple possible auxiliary hyperbolic differential
operators that can serve the same purpose. Here are the key required
properties for such an operator $h$: (a) $h$ must be a cochain map that is
homotopic to zero with respect to the de~Rham complex, (b) it must possess
retarded and advanced Green functions, (c) these Green functions must be
causal with respect to the given conal structure on $M$. In fact, the
conclusion of our Theorem~\ref{thm:sc-tc} was reached independently in
the recent paper~\cite{benini} by following an argument structurally
similar to ours, with the d'Alembertian replaced by the Lie derivative
$\Lie_v$ with respect to a complete timelike vector field $v$. It is
clearly (Green) hyperbolic~\cite{bgp,baer,kh-big,kh-peierls} with Green
functions simply given by integration (into the future or past) along
the flow lines of $v$. Moreover, it is cochain homotopic to zero because
of the well known magic formula of Cartan: $\Lie_v = \iota_v \d + \d
\iota_v$.

\subsection{Functoriality}\label{sec:funct}
Recall that ordinary de~Rham cohomology is defined on any finite
dimensional manifold and the pullback of differential forms along a
smooth map between manifolds induces a map between their cohomologies
(in the direction opposite the original smooth map). This observation
has the following well-known formalisation: de~Rham cohomology in degree
$p$, $H^p(-)$, is a contravariant functor%
	\footnote{We shall not delve here into the details of category theory.
	It suffices to say that any statement that we shall make involving
	functors and categories will be simply a very terse transcription of
	some other property that will be spelled out in more elementary terms.
	More details about the functorial properties of de~Rham cohomology can
	be found in~\cite{bott-tu}.} %
from the category of smooth manifolds to the category of real vector
spaces. The same cannot be said for de~Rham cohomology with compact
supports, $H_c^p(-)$, because the pullback of a compactly supported form
need not be compactly supported itself. This pullback problem is fixed
by considering only proper%
	\footnote{A continuous map is \emph{proper} if the preimage of any compact
	set is compact.} %
smooth maps between manifolds. So, given a proper smooth map $f\colon
M\to N$, pullback along it induces a contravariant map between de~Rham
cohomologies in degree $p$ with compact support, $f^*\colon H^p_c(N) \to
H^p_c(M)$. If the map $f$ satisfies a different restrictive condition,
namely that it is an open embedding, it is possible to define a
covariant pushforward map $f_*\colon H^p_c(M) \to H^p_c(N)$: we can
identify $M$ with its image $f(M)$, an open subset of $N$, and extend by
zero any compactly supported form defined $M$ to all of $N$. In short,
de~Rham cohomology with compact supports, $H_c^p(-)$, defines a
contravariant functor on the category of smooth manifolds with proper
maps as morphisms, when paired with the pullback, while it defines a
covariant functor on the category of smooth manifolds with open
embeddings as morphisms, when paired with the pushforward.

A natural question is the following: do similar properties hold, and
under what precise conditions, for de~Rham cohomologies with causally
restricted supports? For instance, this question was briefly raised, but
without any definite answer, in~\cite{benini}. In fact, it is straight
forward to present causally restricted cohomologies as functors,
provided we modify the domain category by adding generalized causal
structures to manifolds (as in Section~\ref{sec:gen-caus}) and by
modifying the notion of a proper map with respect to the causal
structure.

Consider two conal manifolds $M$ and $N$, with a smooth map $f\colon
M\to N$ between them. We call the map $f$ \emph{reflectively
spacelike-proper} if the preimage of any spacelike compact set is also
spacelike compact, while we call it \emph{reflectively timelike-proper}
if the preimage of any timelike compact set is also timelike compact.
When the map $f$ is an open embedding, we also introduce the terminology
\emph{monotonically spacelike-proper} for the case when the image of any
spacelike compact set is itself spacelike-compact and
\emph{monotonically timelike-proper} for the case when the image of any
timelike compact set is timelike compact. We should note that the above
terminology is partly inspired by some general notions from the theory
of partially ordered sets. A map $f\colon M\to N$ between two partially
ordered sets $(M,\le)$ and $(N,\le)$ is said to be \emph{monotonic} if
$x \le y$ implies $f(x) \le f(y)$ and, on the other hand, it is said to
be \emph{order-reflecting} if $f(x) \le f(y)$ implies $x\le y$. The
following theorem is a straight forward generalization of the previous
arguments for the simpler case of compact supports.
\begin{thm}
Let $\CMan_{\SC}$ and $\CMan_{\tc}$ be the categories of conal manifolds
with, respectively, reflectively spacelike-proper and reflectively
timelike-proper, smooth maps as morphisms, while the $\CMan_{\SC}^{e}$
and $\CMan_{\tc}^{e}$ categories have, respectively, monotonically
spacelike-proper and monotonically timelike-proper open embeddings as
morphisms. Then, de~Rham cohomologies with spacelike and timelike
supports, $H^p_{\SC}(-)$ and $H^p_{\tc}(-)$, are contravariant functors on
$\CMan_{\SC}$ and $\CMan_{\tc}$, respectively. Similarly, $H^p_{\tc}(-)$
and $H^p_{\SC}(-)$ are covariant functors on $\CMan_{\tc}^{e}$ and
$\CMan_{\SC}^{e}$, respectively.
\end{thm}
\begin{proof}
The proof is a direct parallel of the above arguments for the case with
compact supports, since the definitions have been specifically adapted
to that argument.
\end{proof}

To show that the definitions of spacelike- and timelike-proper maps are
in some sense natural, we give a couple of examples.

\begin{lem}
Let $M$ be a manifold and two conal structures on it, $C\sse C' \sse TM$
($C$ is ``slower'' than $C'$) (Section~\ref{sec:gen-caus}). The identity
map is a reflectively spacelike-proper from $(M,C')$ to $(M,C)$ and
reflectively timelike-proper from $(M,C)$ to $(M,C')$.
\end{lem}
\begin{proof}
Let $K\sse M$ be any compact subset. Then, by hypothesis, the
$C$-influence set is smaller than the $C'$-influence set,
$\overline{I}_C(K) \sse \overline{I}_{C'}(K)$. Therefore, any
$C$-spacelike compact set is also $C'$-spacelike and hence the identity
from $(M,C')$ to $(M,C)$ is reflectively spacelike-proper. On the other
hand, if $U\sse M$ is $C'$-timelike compact, then we have the inclusion
$\overline{I}_C(K) \cap U \sse \overline{I}_{C'}(K) \cap U$, the latter
being compact. Therefore, $U$ is also $C$-timelike compact and the
identity from $(M,C)$ to $(M,C')$ is reflectively timelike-proper.
\end{proof}

\begin{lem}
Let $(M,g)$ and $(N,h)$ be two globally hyperbolic Lorentzian manifolds
and $f\colon M\to N$ an open isometric embedding, such that the image of
a Cauchy surface of $M$ is a Cauchy surface of $N$. Then, $f$ is
monotonically timelike-proper.
\end{lem}
\begin{proof}
Let $U\sse M$ be timelike compact. According to~\cite{sanders}, this is
equivalent to $U$ being contained between two Cauchy surfaces in
$(M,g)$, say $\Sigma_1, \Sigma_2 \sso M$. This means that the image,
$f(U)$ is contained between $f(\Sigma_1)$ and $f(\Sigma_2)$, with the
latter, by hypothesis, being Cuachy surfaces in $(N,h)$. Thus, $f(U)$ is
also timelike compact and the map $f$ is monotonically timelike-proper.
\end{proof}

\subsection{Other differential complexes}\label{sec:other}
Our interest in computing the de~Rham and Calabi cohomologies with
causally restricted supports has was motivated by their importance in
understanding the geometric structure of classical and quantum field
theories~\cite{dl,sdh,bds,fl,benini,hs,kh-big,kh-peierls}. Namely, for a
general class of linear field theories, one can formulate sufficient
conditions for the non-degeneracy of the theory's Poisson structure and
the completeness of compactly supported smeared fields as physical
observables in terms of the cohomologies of corresponding differential
complexes. Non-linear field theories can be studied in terms of their
linearizations about arbitrary background solutions. To Maxwell
electrodynamics corresponds the de~Rham
complex~\cite[Sec.4.2]{kh-peierls}. To linearized gravity on constant
curvature backgrounds, corresponds the Calabi
complex~\cite[Sec.4.4]{kh-peierls}. Similarly, to Yang-Mills linearized
about a flat connection corresponds a twisted de~Rham complex.

Each of these examples can be treated using the methods presented in
this paper. Few other explicit examples of differential complexes
corresponding to other field theories of physical interest seem to be
known. In particular, they do not seem to be known for linearized
gravity on non-constant curvature backgrounds and, perhaps, not even for
Yang-Mills linearized about non-flat connections. On the other hand,
there are strong abstract reasons to believe that such differential
complexes do indeed exist~\cite{quillen,goldschmidt-lin,pommaret}.

If such a differential complex also shares the apparently crucial
property of admitting cochain homotopies that generate hyperbolic and
elliptic cochain maps (cf.\ the $E_l^g$, $P_l^g$, $E_l^{g_R}$ and
$P_l^{g_R}$ maps of Sections~\ref{sec:calabi-ops}
and~\ref{sec:calabi-coh}), then its causally restricted cohomologies can
be related to those with unrestricted and compactly supported ones, as
in Theorems~\ref{thm:sc-tc} and~\ref{thm:calabi}.

If, in addition, such a differential complex could also be seen as
resolving a locally constant sheaf, its unrestricted cohomologies could
be computed by algebraic means, without actually solving complicated
systems of differential equations, as in Section~\ref{sec:calabi-coh}.
The latter requirement is closely related to the initial differential
operator in the complex having only a finite dimensional space of
solutions (being of \emph{finite type}), as is the case for the
\emph{locally constant} (de~Rham) and \emph{Killing} (Calabi)
conditions.

The compactly supported cohomologies could also be obtained if the
corresponding formally adjoint complex satisfied similar requirements,
as illustrated in Section~\ref{sec:calabi-coh} by the appearance of the
locally constant sheaf of Killing-Yano tensors.

\section{Discussion}\label{sec:discuss}
We have shown how to compute the de~Rham cohomology with causally
restricted supports (retarded, advanced, past compact, future compact,
spacelike compact and timelike compact) on a globally hyperbolic
Lorentzian spacetime, using special properties of the d'Alembert wave
operator and its Green functions.  The result (Theorems~\ref{thm:pf}, \ref{thm:sc-tc}
and Corollary~\ref{cor:pairing}) expresses these causally restricted
cohomologies in terms of the standard de~Rham cohomologies of the
spacetime manifold, with either unrestricted or compact supports. These
results, confirm the independent similar results of the
recent work~\cite{benini}. However, since our method does not rely on
the strong invariance properties of the de~Rham complex under
topological homotopies, we have also obtained further results. In
particular, our method is also applicable to the Calabi complex
(Theorem~\ref{thm:calabi}). The Calabi complex appears in linearized
gravity on constant curvature backgrounds in a way similar to the
de~Rham complex in Maxwell theory. These results answer some questions
that have naturally arisen in recent investigations of classical and quantum
gauge theories on curved spacetimes.

Finally, we have also made comments about other questions that have
naturally appeared in these investigations. Namely, we discussed the
covariance of causally restricted cohomologies under specific types of
morphisms between spacetimes, adapted to their causal structure, and
under changes of the causal structure itself.

We have presented almost the bare minimum of information about the
Calabi complex that is needed to obtain our results. A fuller discussion
of this interesting complex, including relevant geometric properties
that are difficult to locate in or are absent from the current literature,
is deferred to future work~\cite{kh-calabi}. In the future, it will also
be interesting to find the analogs of the Calabi complex on more general
Lorentzian backgrounds, which would consist of differential complexes
resolving the sheaf of Killing vectors on a given background. However,
we conjecture that the Hodge-like structure that we have used to compute
causally restricted cohomologies will be shared by all of them.

\section*{Acknowledgments}
The author would like to thank Marco Benini for useful discussions and
also Romeo Brunetti, Claudio Dappiaggi, and Valter Moretti for their
support in the course of this work.

\bibliographystyle{utphys-alpha}
\bibliography{paper-causcohom}

\end{document}